\documentclass[prb,nobalancelastpage,twocolumn,superscriptaddress]{revtex4-1}
\usepackage{color,amsthm,amsmath,amsxtra,amsfonts,dsfont,graphicx,bm}
\usepackage[english]{babel}
\usepackage{color}
\usepackage{hyperref}
\usepackage{tikz}

\def\Id{{\openone}}
\newcommand{\be}{\begin{equation}}
\newcommand{\ee}{\end{equation}}
\newcommand{\bea}{\begin{eqnarray}}
\newcommand{\eea}{\end{eqnarray}}

\newcommand{\tr}{\mathrm{tr}}

\newcommand{\nocomma}{}
\newcommand{\tmop}[1]{\ensuremath{\operatorname{#1}}}
\newcommand{\tmscript}[1]{\text{\scriptsize{$#1$}}}

\newtheorem{lemma}{Lemma}
\newtheorem{proposition}{Proposition}
\newtheorem{theorem}{Theorem}

\newcommand{\norm}[1]{\left\| #1 \right\|}
\newcommand{\id}{\mathrm{Id}}
\newcommand{\ep}{\varepsilon}

\begin{document}

\title{Approximating Gibbs states of local Hamiltonians efficiently with PEPS}

\author{Andras \surname{Molnar}}
\affiliation{Max-Planck-Institut f{\"{u}}r Quantenoptik,
Hans-Kopfermann-Str.\ 1, D-85748 Garching, Germany}
\author{Norbert Schuch}
\affiliation{JARA Institute for Quantum Information, RWTH Aachen University, D-52056 Aachen, Germany}
\author{Frank Verstraete}
\affiliation{Faculty of Physics, University of Vienna, Boltzmanngasse 5, 1090 Vienna, Austria and\\ Department of Physics and Astronomy, Ghent University, Ghent, Belgium}
\author{J.~Ignacio \surname{Cirac}}
\affiliation{Max-Planck-Institut f{\"{u}}r Quantenoptik, Hans-Kopfermann-Str.\ 1, D-85748 Garching, Germany}

\begin{abstract}
We analyze the error of approximating Gibbs states of local quantum 
spin Hamiltonians on lattices with Projected Entangled Pair States 
(PEPS) as a function of the bond dimension ($D$), temperature 
($\beta^{-1}$), and system size ($N$). First, we introduce a compression 
method in which the bond dimension scales as 
$D=e^{O(\log^2(N/\epsilon))}$ if $\beta<O(\log (N))$. Second, building 
on the work of Hastings\cite{Hastings2006}, we derive a polynomial 
scaling relation, $D=\left(N/\epsilon\right)^{O(\beta)}$. This implies 
that the manifold of PEPS forms an efficient representation of Gibbs 
states of local quantum Hamiltonians. From those bounds it also follows 
that ground states can be approximated with $D=N^{O(\log(N))}$ whenever 
the density of states only grows polynomially in the system size. All 
results hold for any spatial dimension of the lattice.
\end{abstract}

\maketitle

\section{Introduction}
\label{Introduction}

Problems dealing with quantum many-body systems in lattices appear very often in different branches of Physics and Chemistry. They typically correspond to discretized versions of first-principle continuum models, like in high-energy physics, atomic physics, or quantum chemistry, or provide a phenomenological description of a complex system, as in condensed matter physics. They are characterized in terms of a lattice Hamiltonian, $H$, which describes the motion, as well as the interactions among the different constituents. Apart from generating the dynamics via the Schr\"odinger equation, the Hamiltonian defines the quantum state of the system in thermal equilibrium through the Gibbs density operator,
 \be
 \label{rho}
 \rho = \frac{e^{-\beta H}}{Z} = \frac{e^{-\beta H}}{\tr\left[e^{-\beta H}\right]},
 \ee
where $Z$ is the partition function and $\beta=1/\kappa_B T$ is the inverse temperature (we set the Bolzmann constant $\kappa_B=1$). This operator encodes all the (statical) physical properties of our systems. Extracting that information becomes a hard problem, even for systems consisting of very few particles. The reason is that, in order to determine expectation values of observables, we have to express $\rho$ in a basis of the corresponding Hilbert space, and the dimension of the latter grows exponentially with the number of lattice sites, $N$ (i.e. volume) of the lattice. This fact is ultimately related to the tensor product structure inherent in quantum mechanical problems dealing with composite objects, and thus ubiquitous in several branches of science.

There exist different ways around that problem, at least in some specific situations. For instance, one can employ sampling techniques in certain models (not suffering from the sign problem), to accurately determine the physical properties of a system in thermal equilibrium.  Alternatively, one can restrict oneself to simple tractable families of states depending on few parameters, which can then be determined by variational techniques. This last approach typically requires a good intuition to select which family will encompass all the physical properties that one has to describe, and can easily lead to either wrong or inaccurate results. 
Yet another approach is that of quantum simulation, where the
Hamiltonian of interest is implemented on a different system on which
one has enough control \cite{Cirac2012a}.

Strictly speaking, the exponential scaling of the dimension of the Hilbert space with the size of the lattice should not be the ultimate reason for the difficulty of quantum many-body problems, at least for the ones that naturally appear in nature. For instance, if $H$ is the sum of terms acting non-trivially only on at most $x$ lattice sites, then we can characterize all possible Hamiltonians with a number of parameters that scales only polynomially with $N$. If those terms are local, meaning that the distance between the sites on which term of $H$ acts is bounded by a constant, this scale is even linear in $N$. Thus, for all those problems, $\rho$ itself only depends on few parameters. One says that the states can only explore a very small ``corner'' of the Hilbert space \cite{Cirac2009}. Consequently, it may be possible to utilize this fact to find families of states that {\em describe all possible many-body lattice problems with $x$-body interactions in thermal equilibrium, and that depend on a number of parameters that only grows polynomially with} $N$. Thus, a central problem in this context is to find and characterize such a family of states. A first and fundamental step would be to solve that problem for local Hamiltonians, on which we will concentrate in the following.

Matrix Product States (MPS) \cite{Fannes1992,Perez-Garcia2007} provide the answer for one dimensional models at zero temperature for both, gapped \cite{Hastings2007,Landau2013} and critical models \cite{Verstraete2006}. Specifically, if $\Psi_0$ is the ground state of such a Hamiltonian there exists a MPS of bond dimension $D$, $\Psi_{\mathrm{MPS}}$, such that $\|\Psi_0 -\Psi_{\mathrm{MPS}}\|<\epsilon$ with $D=O[poly(N/\epsilon)]$.
Note that, in turn, the number of parameters to characterize the MPS scales polynomially with $D$. This result is strongly connected to the area law \cite{Srednicki1993,Eisert2010}, which is fulfilled (or only slightly violated) for those models and MPS. In higher dimensions and still at zero temperature, it is conjectured (and proven under certain assumptions \cite{Masanes2009,Hamza2009b}), that the area law still holds (with logarithmic corrections for certain critical models \cite{Wolf2006,Gioev2006}). In that case, one would expect that the Projected Entangled-Pair States (PEPS) \cite{Verstraete2004,Verstraete2004c}, which extend MPS to higher dimensions, would provide us with the efficient description of that corner of the Hilbert space \cite{Cirac2009}. Moreover, for any finite temperature (independent of $N$), an area law has been proven \cite{Wolf2008} both for Gibbs states (\ref{rho}), as well as for Projected Entangled-Pair Operators (PEPO), the extension of PEPS to mixed sates. This also suggests that PEPOs can efficiently describe Gibbs states of local Hamiltonians. 
From the physics point of view, this is actually the relevant question, as any extended system can only be cooled down to a certain temperature independent of the system size.

Hasting \cite{Hastings2006} has already derived some remarkable results addressing that question. He has shown that in $d$ spatial dimensions, one can build a PEPO, $\rho_{\rm PEPO}$, such that $||\rho-\rho_{\rm PEPO}||_1< \epsilon$ with bond dimension scaling as
 \be
 D=e^{O(\beta\log(N/\epsilon)^{d})}.
 \ee
 This gives a polynomial scaling for one dimension, and a sub-exponential (although superpolynomial) one for higher ones. This result also implies a bound for the approximation of the ground state. In fact, if $H$ is gapped and the density of states for a fixed energy only grows as poly$(N)$, then choosing $\beta=O(\log N)$ in (\ref{rho}) we obtain a state that is as close as we want to the ground state \cite{Hastings2007b}. This means that, under those conditions, we can find a PEPS approximation of the ground state with
 \be
 D=e^{O(\log(N/\epsilon)^{d+1})}.
 \ee
In the present paper we derive the following results. First, we use a novel method to obtain a bound for $\beta \leq O(\log(N))$ independent of the dimension (although still superpolynomial in $N$), 
 \be
 \label{Trotterbound1}
 D=e^{O(\log^2(N/\epsilon))}.
 \ee
Under the same condition on the density of states as before, we also obtain that the ground state can be approximated with
 \be
 \label{Trotterbound2}
 D=e^{O(\log^2(N/\epsilon))},
 \ee
independent of the dimension.
Finally, using Hastings' construction of the PEPO (see also \cite{Kliesch2013}), we show that it is possible to have a polynomial scaling for any temperature, i.e.
 \be
 \label{Poly}
 D=(N/\epsilon)^{O ( \beta)}.
 \ee

The paper is organized as follows. In section \ref{problem} we define the problem we are addressing in this work. Section \ref{sec:trotter} derives the bounds (\ref{Trotterbound1}) and (\ref{Trotterbound2}) using a technique based on the Trotter expansion. In Section \ref{sec:Hastings} we use a different encoding of the PEPO based on Hastings' construction to obtain the polynomial bound (\ref{Poly}).
In all these sections we quote the results and explain how we have proven them. In the appendix we give details of the proofs.

\section{Problem}
\label{problem}

We consider a growing sequence of finite spin systems, $S_n$, with two-body interactions. To every system, $S_n$, we assign a graph, $\mathcal{G}_n = (\mathcal{V}_n, \mathcal{E}_n)$, where the vertices $\mathcal{V}_n$ correspond to the individual spins and the edges $\mathcal{E}_n$ to interactions. The Hamiltonian is such that only the connected points interact:
\begin{equation}
  H_n = \sum_{e \in \mathcal{E}_n} h_e,
\end{equation}
where $h_e$ acts non-trivially on spins $v$ and $w$ if $e = ( v, w)$. Even though for simplicity we have considered only nearest neighbor interactions, the results generalize to local more-body interactions. We will assume that the (operator) norm of all the terms in the Hamiltonians is bounded by 1, i.e., $\norm{h_i}\leq 1$. If the norm of the Hamiltonians were bounded by $J$ instead of $1$, this factor could be included into the definition of the temperature.

We assume that all graphs are connected, and that their degree is uniformly bounded. That is, the number of edges starting from a given point is smaller than some constant $z$. This implies that $2|\mathcal{E}_n|/z < |\mathcal{V}_n| \leq |\mathcal{E}_n| + 1$. Thus, we can equally characterize the size of the system by the number of spins or interactions, $N=|\mathcal{V}_n|$ and $|\mathcal{E}_n|$, respectively. For convenience we will denote $|\mathcal{E}_n|$ by $K$ and omit the index $n$ in the following.

We also assume that there is a uniformly bounded lattice growth constant. This means that there is a universal constant, $\gamma$, such that for any given $e \in E$ and all $l \in \mathbb{Z}^+$
\begin{equation}\label{eq:latticegrowth}
  \big| \left\{ \mathcal{I} \subseteq \mathcal{E}| \mathcal{I} \ \text{connected}, \  e \in \mathcal{I}, \  \nocomma | \mathcal{I} | = l \right\}\big| \leq \gamma^l.
\end{equation}
That is, the number of connected regions having $l$ edges that include a specific edge, $e$, grows at most exponentially with $l$. In particular, this is the case if $\mathcal{G}_n$ is a regular lattice in any spatial dimension \cite{Klarner1967}. Thus, our treatment includes all those cases.

We consider the Gibbs state corresponding to $H$ given by  (\ref{rho}). We will construct a PEPO, $\tilde{\rho}$, of bond dimension $D$, that is close to that state. In particular, for any $\varepsilon>0$,
\begin{equation}
  \left\| e^{-\beta H} - \tilde{\rho} \right\|_1 \le \varepsilon \|e^{-\beta H}\|_1,
\end{equation}
where $\| x \|_p=\left[{\rm tr} (x^\dagger x)^{p/2}\right]^{1/p}$ stands for the Schatten-p-norm ($\| x \|=\|x\|_\infty$ for the operator norm). We will be interested in how $D$ scales with $N$ (or equivalently, with $K$) and $\varepsilon$.

By a PEPO on a graph $\mathcal{G} = ( \mathcal{E}, \mathcal{V})$ we mean that the operator
$\tilde{\rho}$ admits the following form:
\begin{equation}
  \tilde{\rho} = \sum_{\alpha : \mathcal{E} \rightarrow \{ 1 \ldots D \}} \bigotimes_{v
  \in \mathcal{V}} X^v_{\alpha ( e^v_1) \ldots \alpha ( e^v_{z ( v)})}. \label{eq:mpo}
\end{equation}
Here, $X^v_{\alpha ( e^v_1) \ldots \alpha ( e^v_{z ( v)})}$ are operators acting on the vertex $v$ alone, $z( v)$ is the degree of $v$, and $e^v_1, \ldots e^v_{z ( v)}$ are the edges going through $v$. This definition is the straightforward generalization of PEPS \cite{Verstraete2004} for operators \cite{Verstraete2004b,Zwolak2004}. One can readily see \cite{Cirac2009} that this operator can be written as a tensor network on the graph $\mathcal{G}$, where the bond dimension is $D$.

\section{Construction based on a Trotter expansion and compression}\label{sec:trotter}

In this section we use a Trotter expansion combined with a compression method to approximate the Gibbs state. The intuition about why this expansion should give rise to a PEPO description is the following (see also \cite{Hastings2010}). Let us assume that the operators $h_e$ commute with each other. Then, the Gibbs state (\ref{rho}) is proportional to a product of exponentials, each of them of the form $e^{-\beta h_e}$. One can easily show that each term in that product creates a link in the PEPO \cite{Wolf2008}. The bond dimension, $D_0$, is simply the maximum number of singular values of $h_e$, when decomposed in terms of the vertices it connects, and thus it is independent of $K$ and the temperature. In the general case where the $h_e$ do not commute with each other, we can still perform a Trotter expansion and approximate $\rho$ (up to a constant factor) by $(\tau^\dagger \tau)^M$ where
\begin{equation}
	\tau = \prod_{i=1}^K e^{-\beta h_i/2M} .\label{eq:tau}
\end{equation}
The integer $M$ has to be chosen such that the approximation is good, i.e.
\begin{equation}
\label{Trotterapp}
	\|e^{-\beta H} -(\tau^\dagger \tau)^M\|_1\leq \varepsilon \|e^{-\beta H}\|_1
\end{equation}
for some $\varepsilon>0$. Now, if we use the same argument we see that each time we apply $\tau$, we create a bond between each pair of vertices that are connected in the graph. That is, we multiply the bond dimension by $D_0$. Thus, naively, the final bond dimension will be $D_0^{2M}$, and since $M$ has to grow polynomially with $K$, we get a very bad bound. However, for large $M$ each of the terms in $\tau$ is close to the identity operator. Thus, this operator creates very little entanglement and it should be possible to compress the information that is contained in the bond variables for any pair of connected vertices, and therefore to decrease the bond dimension. In fact, in the case of commuting Hamiltonians one can reduce it to $D_0$, independent of $M$. This is, in fact, what we do in this section: we first find $M$ such that (\ref{Trotterapp}) holds, and then we compress the bond to get a better scaling of the bond dimension with $K$.

More specifically, we write $e^{-\beta h_i/2M}=\Id + (e^{-\beta h_i/2M}-\Id)$, then, after collecting the $K$ terms of $\tau$ and $\tau^\dagger$ into one product of $2K$ terms, we obtain
\begin{equation}
(\tau^\dagger \tau)^M=\prod_{j = 1}^M \prod_{i = 1}^{2K} e^{- \beta \tilde{h}_i / 2 M} = \prod_{j = 1}^M \prod_{i = 1}^{2K} ( 1 + x_i),
\end{equation}
where $\tilde{h_i}$ denotes $h_{K+1-i}$ if $i\leq K$, and $h_{i-K}$ otherwise, and $x_i=e^{-\beta \tilde{h}_i/2M}-\Id$. After expanding the product, this operator takes the form
\begin{equation}\label{eq:expansion}
(\tau^\dagger \tau)^M= \sum_{\lambda
  \in \mathcal{M}_{M,2K}^b} \prod^M_{j = 1} \prod_{i = 1}^{2 K}
  x_i^{\lambda_{i, j}}.
\end{equation}
The sum runs over all $M\times2K$ matrices with entries $0$ or $1$, denoted by $\mathcal{M}_{M,2K}^b$.
From this sum we only keep those terms in which any given $x_i$ appears at most $L$ times in (\ref{eq:expansion}).  In Section \ref{sec:compression} we show that the resulting operator $\tilde{\rho}$ is a good approximation to $(\tau^\dagger \tau)^M$ if $L \approx \log K$.

In section \ref{sec:coding} we show then that the resulting operator can be written as a PEPO, in the sense of (\ref{eq:mpo}), with bond dimension $M^{O ( L)}$. The reason why this operator admits a PEPO form can be understood as follows. First we identify each particular term in the expansion of $(\tau^\dagger \tau)^M$ with the help of indices defined on the edges. This can be done by specifying at every edge, $i$, the position where $x_{K+1-i}$ and/or $x_{K+i}$ appear out of the $M$ possibilities. Once a term is identified, we proceed with the Schmidt decomposition of that term in order to build the local operators $X^v_{\alpha ( e^v_1) \ldots \alpha ( e^v_{z ( v)})}$. Let us notice that the latter only depends on the order in which the operators $x_{e_1^v}$, $x_{e_2^v}$, ... $x_{e_{z(v)}^v}$ appear in the given term, where $e_1^v,\dots e_{z(v)}^v$ are the edges starting from point $v$. This order can be obtained locally from the edges that surround $v$, which contain information about the $x$ involved in each of them.
As a result of that, at every edge we have to specify ${M \choose L}^2\approx M^{2L}$ natural numbers. As $M=poly(K)$ and $L=O(\log K)$, this gives a bond dimension $K^{O(\log K)}$ for the approximating operator. Therefore, as $N\leq 2K/z$, we obtain a bond dimension that scales like $N^{O(\log N)}$.

\subsection{Trotter expansion}\label{sec:trotterproof}

 We know that $(\tau^\dagger \tau)^M$ ($\tau$ as in equation (\ref{eq:tau})) tends to $e^{-\beta H}$ if $M\to \infty$. The question is how big M has to be chosen such that we obtain a good approximation in one-norm. Here we prove  that setting $M=poly(K)$ is enough  (see also \cite{Berry2006}).

We present the proof in two steps. First we show that $\|e^{- \beta H}-(\tau^\dagger \tau)^M\|_1$ is small compared to $\|e^{-\beta H}\|_1$ as long as $\|\eta-\tau\|_{2M}$ is small compared to $\|\eta\|_{2M}$ where $\eta=e^{- \beta H / 2 M}$. Second, we show that $\|\eta-\tau\|_{2M}$ is small compared to $\|e^{-\beta H/2M}\|_{2M}$. The key point is that both $e^{-\beta H}$ and $(\tau^\dagger\tau)$ are close to $(\Id-\beta H/M)^M$.
We state the first step as a proposition:

\begin{proposition}
  \label{prop:T} If $\ep<1/3$ and
  \begin{equation}
    \| \eta - \tau \|_{2 M} \leq \frac{\varepsilon}{M} \| \eta \|_{2M},
  \end{equation}
  then
  \begin{equation}
    \| \eta^{2 M} - ( \tau^{\dagger} \tau)^M \|_1 \leq
    9 \varepsilon  \| \eta^{2 M} \|_1.
  \end{equation}
\end{proposition}

The proof combines the identity $a^m-b^m=\sum_i a^i (a-b) b^{m-i-1}$ with the H\"older inequality for matrices \cite{Bhatia1997} and it is presented in the Appendix. We state the second statement (that $\eta$ is close to $\tau$) as a lemma.

\begin{lemma} \label{lem:trotterproof}
 	If $M > 36 \beta^2 K^2/\epsilon$ and $\epsilon < 1$, then
  \[ \| \eta - \tau \|_{2 M} \leq \frac{\epsilon}{M} \| \eta \|_{2 M}. \]
\end{lemma}

The main idea is that it is enough to prove the statement for the operator norm, as $\|\eta -\tau\|_{2M}$ is bounded by the H\"older inequality
$$\|\eta-\tau\|_{2M} = \|\eta^{-1}\eta(\eta-\tau)\|_{2M}\leq \|\eta^{-1}\| \|\eta\|_{2M}\|\eta-\tau\|,$$ 
and $\|\eta^{-1}\|$ is not too big as $\eta$ is close to the identity operator. In order to show that $\|\eta-\tau\|$ is close to zero, by a simple series expansion we obtain that $\|\eta-\Id+\beta H/M\|$ is small and so is $\|\tau-\Id+\beta H/M\|$. The statement then follows from the triangle inequality. The detailed proof is presented in the appendix.

Putting together Proposition (\ref{prop:T}) and Lemma (\ref{lem:trotterproof}), we obtain that the Trotter approximation is $\ep$-close (in one-norm) if the trotter steps are chosen to be $M>360\beta^2 K^2/\varepsilon$.

\subsection{Compression}\label{sec:compression}
We approximate now $( \tau^{\dagger} \tau)^M$ by an operator $\tilde{\rho}$ starting from  Eq. (\ref{eq:expansion}). This expansion can be pictured as follows. We can think of the resulting operator as a sum:
\begin{equation}
  ( \tau^{\dagger} \tau)^M = \sum_{\text{all fillings}}
  \begin{array}{|c||c|c|c|c|}
    \hline
    & x_1 & x_2 & \ldots & x_{2 K}\\
    \hline \hline
    1 & X &  &  & \\
    \hline
    2 &  &  & X & \\
    \hline
    \vdots &  & X &  & \\
    \hline
    M & X &  &  & X\\
    \hline
  \end{array}\ ,
\end{equation}
where the table can be understood as follows. We begin to read from the upper-left
corner, from left to right, row-by-row. Whenever we meet an $X$ in the actual
cell, we write down the corresponding operator $x_i$ (according to the column), and
otherwise the identity operator. The value assigned to a given table is then the product of those operators. We finally have to sum up the resulting operators for all possible fillings
of the table.

The approximating operator $\tilde{\rho}$ can be thought of in the same way, just
limiting the number of $X$'s in each of the columns.
\begin{equation}
  \tilde{\rho} = \sum_{\tmscript{\begin{array}{c}
    \text{filling per}\\
    \text{column}  \leq L
  \end{array}}} \begin{array}{|c||c|c|c|c|}
    \hline
    & x_1 & x_2 & \ldots & x_{2 K}\\
    \hline\hline
    1 & X &  &  & \\
    \hline
    2 &  &  & X & \\
    \hline
    \vdots &  & X &  & \\
    \hline
    M & X &  &  & X\\
    \hline
  \end{array}\ .
\end{equation}
We want to prove that this is a good approximation: $\| ( \tau^{\dagger}
\tau)^M - \tilde{\rho} \|_1 \leq \varepsilon \| \rho \|_1$ if the maximal number of $X$'s per column, $L$, is chosen big enough. We will show that $L = O ( \log K)$ is enough.

Let us first explain the main idea of the proof. Given a set of columns $\mathcal{I}\subseteq \{1,2\dots K\}$, define $S(\mathcal{I})$  to be the sum of all tables containing more than $L$ $X$'s in all columns $i\in \mathcal{I}$, but with no restriction for the columns not belonging to $\mathcal{I}$. Formally, let $\mathcal{Q}(\mathcal{I})$ denote the set of these tables:
\[ \mathcal{Q} ( \mathcal{I}) = \left\{ \lambda \in \mathcal{M}_{M,2K}^b \mid i \in \mathcal{I} \Rightarrow \sum_j \lambda_{i, j} > L \right\} , \]
then $S(\mathcal{I})$ is the sum
\begin{equation}
  S ( \mathcal{I})  =  \sum_{\lambda \in \mathcal{Q} ( \mathcal{I})} \prod^M_{j = 1} \prod_{i = 1}^{2 K}
  x_i^{\lambda_{i, j}}.
\end{equation}
In any column that has no restriction, the sum can be evaluated, giving back $e^{-\beta \tilde{h}_i/2M}$ in every row of that column. By evaluating those sums we arrive to a sum containing only a few terms. In these remaining terms still a large number of $X$'s appear, therefore the norm of each such term is small. Thus the one-norm of $S(\mathcal{I})$ can be bounded. We will express $\tilde{\rho}$ with the help of the sums $S(\mathcal{I})$ in order to be able to bound its norm.

 We use this observation in order to upper bound the one-norm of $(\tau^\dagger \tau)^M - \tilde{\rho}$. That difference contains one or more columns where there are  more than $L$ appearances of $X$. We regroup the tables as follows. First, given a set of columns, $\mathcal{I}$, we sum up all tables that have more than $L$ appearances of $X$ in the columns $i\in\mathcal{I}$, albeit at most $L$ in all columns $i\notin \mathcal{I}$. This set of tables is the following set:
\[ \mathcal{T} ( \mathcal{I}) = \left\{ \lambda \in \mathcal{M}_{M,2K}^b \mid \sum_j \lambda_{i, j} > L
   \Leftrightarrow i \in \mathcal{I} \right\} . \]
The sum of these tables will be called $R(\mathcal{I})$:
\begin{equation}
  R ( \mathcal{I})  =  \sum_{\lambda \in \mathcal{T} ( \mathcal{I})} \prod^M_{j = 1} \prod_{i = 1}^{2 K}
  x_i^{\lambda_{i, j}}.
\end{equation}
Note that the operator $\tilde{\rho}$ is expressed by $R(\emptyset)$, as $\tilde{\rho}$ is the sum of tables that in each column contain at most $L$ $X$'s.

We can express the sum $S(\mathcal{I})$ with the help of $R(\mathcal{I})$:
\begin{equation}
  S ( \mathcal{I}) = \sum_{\mathcal{J} \supseteq \mathcal{I}} R ( \mathcal{J}), \label{eq:S}
\end{equation}
because in any table in  $S(\mathcal{I})$, the columns containing more than $L$ $X$'s form a set $\mathcal{J}\supseteq\mathcal{I}$.
Note that $(\tau^\dagger \tau)^M=S(\emptyset)$, as   $( \tau^{\dagger} \tau)^M$ contains all tables, with no restriction on the number of $X$'s in any column.

The difference $(\tau^\dagger \tau)^M - \tilde{\rho}$ is then
\begin{equation}\label{eq:S-R}
	( \tau^{\dagger} \tau)^M - \tilde{\rho} = S ( \emptyset) - R ( \emptyset).
\end{equation}
To bound the norm of this difference, we need to express $R(\emptyset)$ with the help of the $S(\mathcal{I})$'s; that is, we need the inverse relation of Eq. (\ref{eq:S}). This inverse relation is given by the M{\"o}bius inversion formula, which is used, for example, in the context of  the Kirkwood-Salzburg equations, for a cluster expansion for the partition function \cite{Kotecky1986,Griffiths1980}. The statement of the M{\"o}bius inversion is the following.

	Let $\mathcal{A}$ be a finite set, $\mathcal{P} ( \mathcal{A})$ the set of all its subsets, and $V$ a vector space. Given a function, $f:\mathcal{P} ( \mathcal{A}) \rightarrow V$, we define the following transformations:
  \begin{eqnarray}
    \hat{f} ( \mathcal{I}) & : = & \sum_{\mathcal{J} : \mathcal{A} \supseteq \mathcal{J} \supseteq \mathcal{I}} f ( \mathcal{J}) \label{eq:hat}\\
    \check{f} ( \mathcal{I}) & : = & \sum_{\mathcal{J} : \mathcal{A} \supseteq \mathcal{J} \supseteq \mathcal{I}} ( - 1)^{| \mathcal{J} \backslash    \mathcal{I} |} f ( \mathcal{J}) \label{eq:check} .
  \end{eqnarray}

\begin{lemma}[M{\"o}bius inversion] \label{lem:Mobius}
  \[ \hat{\check{f}} = \check{\hat{f}} = f \]
\end{lemma}

This lemma just expresses that the second transformation is the inverse of the first one. The proof is presented in the Appendix. We will use the lemma by setting $\mathcal{A}$ to be the set of columns, and $f = R$. Thus,  comparing the definitions (\ref{eq:S}) and (\ref{eq:hat}) we deduce that $\hat{f} = S$. Applying the lemma we obtain the desired relation
$$R(\emptyset)=\sum_\mathcal{I} ( - 1)^{| \mathcal{I} |} S ( \mathcal{I}),$$
 and thus substituting back to Eq. (\ref{eq:S-R})
\begin{equation}
  ( \tau^{\dagger} \tau)^M - \tilde{\rho} =  S ( \emptyset) - \sum_\mathcal{I} ( - 1)^{| \mathcal{I} |} S ( \mathcal{I}) ,
\end{equation}
therefore
\begin{equation}
  ( \tau^{\dagger} \tau)^M - \tilde{\rho} = - \sum_{\mathcal{I} \neq \emptyset} (
  - 1)^{| \mathcal{I} |} S ( \mathcal{I}).
\end{equation}
The one-norm of the difference can be bounded by the triangle inequality:
\begin{equation}\label{eq:compression_bound}
\| ( \tau^{\dagger} \tau)^M - \tilde{\rho} \|_1 \leq  \sum_{m = 1}^{2 K} \binom{2 K}{m} \max_{\mathcal{I}:|\mathcal{I}| = m} \|S (\mathcal{I})\|_1.
\end{equation}
We obtained this form by counting the number of subsets $\mathcal{I}$ of the $2K$ columns that have $|\mathcal{I}|=m$. Now, we need to bound the one-norm of $S (\mathcal{I})$. First of all, as noted before, we can sum up over all indices possessing no restriction. That is, over all $\lambda_{i,j}$ with $i\notin \mathcal{I}$. For example, if $2\notin \mathcal{I}$ then
\begin{equation}
  S(\mathcal{I}) = \sum_{\tmscript{\begin{array}{c}
    \text{filling} \leq L\\
    \text{for column } i\in \mathcal{I}
  \end{array}}} \begin{array}{|c||c|c|c|c|}
    \hline
    & x_1 & x_2 & \ldots & x_{2 K}\\
    \hline\hline
    1 & X & e^{-\beta h_2} &  & \\
    \hline
    2 &  & e^{-\beta h_2} & X & \\
    \hline
    \vdots &  & e^{-\beta h_2} &  & \\
    \hline
    M & X & e^{-\beta h_2} &  & X\\
    \hline
  \end{array}\ ,
\end{equation}
where we have already summed up for all $\lambda_{2,j}$. Let $\mu$ be such a term in $S(\mathcal{I})$ in which each $x_i \ (i\in \mathcal{I})$ is appearing exactly $k_i>L$ times.
The one-norm of this term is bounded by the following lemma.
\begin{lemma} \label{lem:S_norm}
If $M > 72 \beta^2 K^2$, then
$$\left\|\mu\right\|_1\leq 3\|e^{-\beta H}\|_1 \left(\frac{3\beta}{M}\right)^{k_1+\dots k_m}.$$
\end{lemma}
 This bound is the consequence of the fact that the $x_i$'s, whose norm is small, appear exactly $k_1+k_2+\dots k_m$ times in $\mu$, while the rest of the operators, that is, $e^{-\beta \tilde{h}_i}$, give almost a Trotter approximation of $e^{-\beta H}$. The proof is presented in Appendix \ref{sec:appD}. The number of such terms $\mu$ is given by 
\begin{equation}
	{M \choose k_1} {M \choose k_2} \dots {M \choose k_m},
\end{equation}
as at each column $i\in \mathcal{I}$ one has to choose $k_i$ rows out of the total number of $M$ rows to place the appearing $x_i$'s.
Thus the one-norm of $S(\mathcal{I})$ is bounded by the following sum:
\begin{equation}
	S(\mathcal{I})\geq\sum_{k_1 >L}\dots \sum_{k_m> L} 3\| e^{- \beta H} \|_1 \prod_{i=1}^m{M\choose k_i} \left( \frac{3\beta}{M} \right)^{k_i} ,
\end{equation}
therefore
\begin{equation}\label{eq:binom_lem_need}
\| S ( \mathcal{I}) \| \leq 3 \| e^{- \beta H} \|_1 \left( \sum_{k > L} \binom{ M}{k}  \left(
   \frac{3\beta}{M} \right)^{k} \right)^m .
\end{equation}
The sum in the parenthesis can be upper bounded by  
$$\sum_{k > L} \binom{M}{k} \left(\frac{3\beta}{M}\right)^k \leq e^{3\beta} \left( \frac{3e\beta}{L} \right)^L$$ 
(see Lemma \ref{lem:binom} in Appendix \ref{sec:appE}) and thus
   \begin{equation}
     \| S ( \mathcal{I}) \| \leq 3 \| e^{- \beta H} \|_1 \left[ e^{3\beta} \left( \frac{3e \beta}{L}
     \right)^L \right]^m.
   \end{equation}
Substituting the obtained bound into Eq. (\ref{eq:compression_bound}) the following holds for the error of the compression:
   \begin{equation}
     \| (\tau^\dagger\tau)^M - \tilde{\rho} \|_1 \leq 3 \| e^{- \beta H} \|_1 \sum_{m = 1}^{2K} \binom{2K}{m} \left[
     e^{3\beta} \left( \frac{3e \beta}{L} \right)^L \right]^m .
   \end{equation}
   Thus, after evaluating the sum, we obtain
   \begin{equation}
      \left\| (\tau^\dagger\tau)^M - \tilde{\rho} \right\|_1 \leq 3 \| e^{- \beta H} \|_1 \left( \left[ 1 + e^{3\beta} \left( \frac{3e \beta}{L} \right)^L \right]^{2K} - 1\right).
   \end{equation}
   As $(1+x/K)^K\leq e^x\leq 1+2x$ as long as $x<1$, this yields the bound
   \begin{equation}
     \| (\tau^\dagger \tau)^M - \tilde{\rho} \|_1  \leq  12 \| e^{- \beta H} \|_1 K e^{3\beta} \left( \frac{3e \beta}{L} \right)^L.
   \end{equation}
Therefore, if $\beta\leq b\log K$, setting $L=O(\log K/\epsilon)$ implies
	\begin{equation}
     \| (\tau^\dagger \tau)^M - \tilde{\rho} \|_1  \leq \epsilon \| e^{- \beta H} \|_1,
   \end{equation}
thus the error of the compression is bounded by $\epsilon$ if $L=O(\log K/\epsilon)$ and $M > 72 \beta^2 K^2$.

\subsection{Coding as a PEPO}\label{sec:coding}

We show that the resulting operator $\tilde{\rho}$ admits a PEPO form as in
equation (\ref{eq:mpo}):
\begin{equation}
  \tilde{\rho} = \sum_{\alpha : \mathcal{E} \rightarrow \{ 1 \ldots D \}} \bigotimes_{v
  \in \mathcal{V}} X^v_{\alpha ( e^v_1) \ldots \alpha ( e^v_{z ( v)})}.
\end{equation}
First, let us consider the Schmidt decomposition of the operators $x_i$.
\begin{equation}\label{eq:schmidt_coding}
x_i=e^{-\beta \tilde{h}_i/2M}-1=\sum_{\nu=1}^s A^{v,i}_\nu \otimes A^{w,i}_\nu,
\end{equation}
with $s$ being at most $d_{phys}^2$, where $d_{phys}$ is the dimension of the Hilbert space describing the individual spins, and the edge corresponding to column $i$ is composed of the two particles $v$ and $w$. Note that there are two columns associated to a Hamiltonian term $h_i$, $K+1-i$ and $K+i$.

After this decomposition, we can think of $\tilde{\rho}$ as the following sum:
\begin{equation}
  \tilde{\rho} = \sum_{\tmscript{\begin{array}{c}
    \text{filling per}\\
    \text{column} \leq L
  \end{array}}} \begin{array}{|c||c|c|c|c|}
    \hline
    & x_1 & x_2 & \ldots & x_{2 K}\\
    \hline\hline
    1 & 3 & 1 & 0 & 0\\
    \hline
    2 & 0 & 0 & 3 & 0\\
    \hline
    \vdots & 0 & 3 & 0 &4 \\
    \hline
    M & 2 & 0 & s & 0\\
    \hline
  \end{array}\ ,
\end{equation}
where the sum runs over all fillings that have at most $L$ cells different from $0$ in every column. The table means the following. We begin to read the table from left to right, row-by-row. Whenever we meet a cell in column $i$ containing the number $k$ we write down the operator $A^{v,i}_k \otimes A^{w,i}_k$ as in Eq. (\ref{eq:schmidt_coding}). Otherwise we write down the identity operator. The value of the table is again the product of these operators. 

Every term in the above sum is now a tensor product. The local operator acting on particle $v$ depends only on the columns corresponding to the edges surrounding $v$. Indeed, operators acting non-trivially on particle $v$ occur only in these columns.

Therefore, the index $\alpha(e)$ at edge $e$ will specify a possible filling of the two columns corresponding to $e$, and the operator  $X^v_{\alpha ( e^v_1) \ldots \alpha ( e^v_{z ( v)})}$ will mean the product of the corresponding Schmidt coefficients.

For a given edge $\alpha(e)$ can take 
\begin{equation}
D=\left[\sum_{k\leq L} {M \choose k} s^k\right]^2\leq L^2 (sM)^{2L}
\end{equation}
different values, as the positions of the non-zero elements and their values are needed to be specified for the two columns corresponding to edge $e$.

In Section \ref{sec:trotter} we have shown that we should set $M>360\beta^2K^2/\epsilon$ in order to the Trotter approximation be $\epsilon$-close to the Gibbs state. In Section \ref{sec:compression} we have seen that one can choose $L$ such that the compressed operator, $\tilde{\rho}$, is $\epsilon$-close to the Trotter expansion. Therefore, by the triangle inequality, for any given $\epsilon$ that decreases at most polynomially in the system size, one can approximate the Gibbs state with error $\epsilon$, if the Trotter steps are taken to be $poly(K)$ and the compression, $L$, to be $O(\log(K))$. Thus, our method gives a PEPO approximation with bond dimension $K^{O(\log (K))}$. As $2K/z\leq N$, this is a PEPO with bond dimension $N^{O(\log (N))}$.

In Section \ref{sec:compression} we only have supposed that $\beta \leq b \log (K)$, or equivalently, $\beta \leq b \log (N)$. If $H$ is gapped and the density of states for a fixed energy only grows as $poly(N)$, then by setting $\beta = O( \log (N))$, the ground state projector is approximated by the Gibbs state with an error decreasing as $poly(N)$.  Therefore, our method also gives an $N^{O(\log (N))}$ bond dimensional PEPO approximation of the ground state projector, and thus an $N^{O(\log (N))}$ bond dimensional PEPS approximation for the ground state (for any prescribed error $\epsilon$ that decreases at most as $poly(N)$) under the same condition.

\section{$\tmop{Poly} (N)$ bond dimensional approximation}\label{sec:Hastings}

In this section we show that with the help of the cluster expansion technique {\cite{Hastings2006}} we can approximate the thermal state by an MPO with $N^{O(\beta)}$ bond dimension. For that, we just have to modify theorem 15 in {\cite{Kliesch2013}} and introduce a more efficient way of encoding the PEPO. That theorem says that for $\beta < \beta^{*}$ ($\beta^{*}$ is a constant) the density operator can be well approximated with the truncated cluster expansion, where only clusters of size at most  $O(\log K)$ (equivalently, $O(\log N)$) are included. 
By a clever choice of the coding of the PEPO, we show that for that temperature one just needs a $poly(N)$ bond dimension, and then, as in \cite{Hastings2006}, we extend the result to lower (but finite) temperatures.

\subsection{Cluster expansion}\label{sec:cluster}

Before restating theorem 15 in {\cite{Kliesch2013}} we need to introduce some
notation. Let $\mathcal{E}^{*} = \cup_{k = 0}^{\infty} \mathcal{E}^k$, that is, a word $w$ from $\mathcal{E}^{*}$ denotes a sequence of edges: $w = ( w_1 w_2 \ldots w_k)$.
Let $h_w$ denote the product of the Hamiltonian terms corresponding to those edges, $h_w = h_{w_1} h_{w_2} \ldots h_{w_k}$, 
and let $\tmop{supp} ( w)$ be the set of all edges occurring in $w$.

Every word's support is a set $\mathcal{I} \subseteq \mathcal{E}$. One can break it into
connected components: $\mathcal{I} = \cup_i \mathcal{I}_i$ where the $\mathcal{I}_i$'s are connected, and different components do not contain common points. These connected components are also called clusters. Then, let
$ \mathcal{W}_{L} \subseteq \mathcal{E}^{*}$ be the set of all words whose support contains only connected components of size at most
$L$. $\beta^{*}$ will denote a constant such that $\alpha e^{( 2 z - 1) \beta^{*}} (
  e^{\beta^{*}} - 1) < 1$, and
  \begin{equation}
    \tilde{\rho} = \sum_{w \in  \mathcal{W}_{L}}
    \frac{( - \beta)^{| w |}}{| w | !} h_w.
  \end{equation}
Theorem 15 in {\cite{Kliesch2013}} contains the following statement:

\begin{theorem}\label{thm:kliesch}
  If $\beta\leq \beta^*$, then
  \begin{equation}
    \| e^{-\beta H} - \tilde{\rho} \|_1 \leq \| e^{-\beta H} \|_1 \cdot \left( \exp \left( K
    \frac{x^L}{1 - x} \right) - 1 \right)
  \end{equation}
  with $x =  \gamma e^{( 2 z - 1) \beta} ( e^{\beta} - 1) < 1$.
\end{theorem}


Similar to equations (56-58) in {\cite{Kliesch2013}} one can show that the
operator $\tilde{\rho}$ admits the following form:
\begin{equation}\label{eq:cluster_approx}
  \tilde{\rho} = \sum_{\tmscript{ \begin{array}{c}
    \mathcal{I} \in \mathcal{C}_L\\
    \mathcal{I} = \uplus \mathcal{I}_i
  \end{array}}} \prod_i \check{f} ( \mathcal{I}_i)
\end{equation}
 where $\mathcal{C}_L$ means the subsets of edges $\mathcal{I}$  that does not contain a connected component of size bigger than $L$, and the connected components of $\mathcal{I}$ are $\mathcal{I}_i$'s. The operators $\check{f} ( \mathcal{I}_i)$ act locally on $\mathcal{I}_i$ and are defined as:
\begin{equation}
  \check{f} ( \mathcal{I}) = \sum_{\tmscript{\begin{array}{c}
    w \in \mathcal{I}^{*}\\
    \tmop{supp} ( w) = \mathcal{I}
  \end{array}}} \frac{( - \beta')^{| w |}}{| w | !} h_w \ .
\end{equation}
We show in Appendix \ref{sec:appF} that $\check{f} ( \mathcal{I})$ is the M{\"o}bius transform of
$f ( \mathcal{I}) = e^{- \beta' H ( \mathcal{I})}$, with $H ( \mathcal{I}) = \sum_{e \in \mathcal{I}} h_e$. This observation makes it
easier to show that $\tilde{\rho}$ admits the form (\ref{eq:cluster_approx}). 

\subsection{Coding}\label{sec:hastingscoding}

We show in this subsection that the truncated cluster expansion $\tilde{\rho}$  (\ref{eq:cluster_approx}) can be written as a PEPO  [cf. Eq. (\ref{eq:mpo})].
This operator has a very special form. It is a sum of products of local operators, such that the operator acting on a vertex $v$ only depends on the cluster where $v$ is contained in. Therefore, coding $\tilde{\rho}$ as a PEPO will be carried out in two steps.
 First, we enumerate all $\mathcal{I}\in \mathcal{C}_L$ subsets of edges with the help of an index $\alpha_1:\mathcal{E}\to \{1,2\dots B_1\}$. This indexing will be such that for any given $v\in \mathcal{V}$ vertex the surrounding edges encode the information in which cluster $v$ is located. Once the cluster $\mathcal{I}_i\ni v$ is identified, the operator $\check{f}(\mathcal{I}_i)$ is written as a PEPO with the help of an index  $\alpha_2:\mathcal{E}\to \{1,2\dots B_2\}$. The index $\alpha$ used at the description of the PEPO is then the composition of $\alpha_1$ and $\alpha_2$ taking $B_1B_2$ different values. 
\paragraph{Identifying the clusters.}  Let the different values of $\alpha_1(e)$ enumerate all clusters containing $e$ and of size at most $L$. For a given cluster size $l$, there are at most $\gamma^l$ clusters containing $e$ (see Eq. \ref{eq:latticegrowth}), therefore there are at most $L\gamma^L$ such clusters. As $L=O(\log K)$, this means that $\alpha_1$ takes at most $B_1\leq poly(K)$ different values. Let us now examine how this indexing is related to the original goal: to enumerate all $\mathcal{I}\in \mathcal{C}_L$ subsets of edges. For any given $\mathcal{I}\in \mathcal{C}_L$ subset one can find the corresponding values $(\alpha_1(e))_{e\in \mathcal{E}}$. However, given an indexing, $\alpha_1$, it might not correspond to such a subset of edges. 
 The reason is the following. Given an indexing $(\alpha_1(e))_{e\in \mathcal{E}}$, each index means a cluster $\mathcal{I}_e$. The subset $\mathcal{I}\in \mathcal{C}_L$ corresponding  to this $\alpha_1$ is $\cup_{e} \mathcal{I}_e$, if for any two edges $e$ and $f$ either  $\mathcal{I}_e= \mathcal{I}_f$, or the two clusters $\mathcal{I}_e$ and $\mathcal{I}_f$ do not have common point. Therefore  the indexing does \emph{not} correspond to an $\mathcal{I}\in \mathcal{C}_L$ subset if and only if there are two edges $e$ and $f$ such that $\alpha_1(e)$ and $\alpha_1(f)$ denote two different, but overlapping clusters. 
Let us join $e$ and $f$ with a path of edges  going in the union of the two clusters $\mathcal{I}_e$ and $\mathcal{I}_f$. Along that path there is a contradiction locally; otherwise, $e$ and $f$ cannot specify contradictory information (see Figure \ref{fig:indexing}). Therefore, if an indexing $\alpha_1$ does not correspond to a subset of edges, then there is a point $v\in \mathcal{V}$ where it can be detected.

    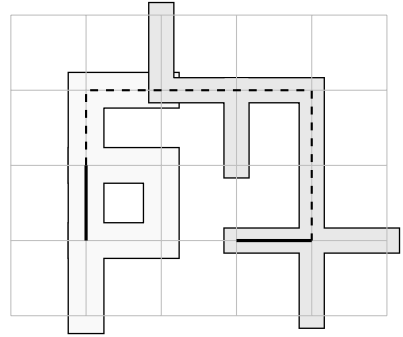
\begin{figure}[h]
    	\centering
		\begin{tikzpicture}
			\draw[line width=14pt,cap=rect] (1,0)--(1,3)--(2,3);
			\draw[line width=14pt,cap=rect] (1,1)--(2,1)--(2,2)--(1,2);
			\draw[line width=13pt,color=gray!5,cap=rect] (1,0)--(1,3)--(2,3);
			\draw[line width=13pt,color=gray!5,cap=rect] (1,1)--(2,1)--(2,2)--(1,2);
			\draw[line width=10pt,cap=rect] (2,4)--(2,3)--(3,3)--(3,2);
		 	\draw[line width=10pt,cap=rect] (3,3)--(4,3)--(4,0);
		 	\draw[line width=10pt,cap=rect] (3,1)--(5,1);
			\draw[line width=9pt,color=gray!18,cap=rect] (2,4)--(2,3)--(3,3)--(3,2);
			\draw[line width=9pt,color=gray!18,cap=rect] (3,3)--(4,3)--(4,0);
			\draw[line width=9pt,color=gray!18,cap=rect] (3,1)--(5,1);
			\draw[thin,color=black!25] (0,0) grid (5,4);
			\draw[very thick] (1,1)--(1,2);
			\draw[very thick] (3,1)--(4,1);
			\draw[thick, dashed] (1,2)--(1,3)--(4,3)--(4,1);
		\end{tikzpicture}
      \caption{Two clusters specified by the thick edges. The information contained in those edges contradict as the clusters overlap. However, the contradiction appear locally somewhere along the dashed line. Thus, our coding will give the 0 operator for this configuration.}
      \label{fig:indexing}
    \end{figure}

\paragraph{Coding the local operators.} Any operator defined on at most $L$ particles can be written as a PEPO with bond dimension $d_{spin}^{2L}$, where $d_{spin}$ is the dimension of the Hilbert space of the particles. For example, an expansion in a product basis of the operators supported on $L$ particles can be viewed as a PEPO. As $\check{f}(\mathcal{I}_i)$ is such a local operator with $L=O(\log K)$, this coding requires an index $\alpha_2$ with  $B_2=poly(K)$ different values. The local operators used for this construction will be $Y^v_{\alpha_2(e_1^v),\dots \alpha_2(e^v_z(v))}(\mathcal{I}_i)$.

With the help of the index $\alpha=(\alpha_1,\alpha_2)$ the operators $X^v_{\alpha ( e^v_1) \ldots \alpha ( e^v_{z ( v)})}$ are constructed as follows. If $\alpha_1(e_1^v),\alpha_1(e_2^v)\dots \alpha_1(e_{z(v)}^v)$ both specify the same cluster $\mathcal{I}_i$  (or some of them the empty cluster, if compatible with $\mathcal{I}_i$), then let 
\begin{equation}
X^v_{\alpha ( e^v_1) \ldots \alpha ( e^v_{z ( v)})}=Y^v_{\alpha_2(e_1^v),\dots \alpha_2(e^v_z(v))}(\mathcal{I}_i),
\end{equation}
 otherwise, if both of them specify the empty cluster, let $X^v_{\alpha ( e^v_1) \ldots \alpha ( e^v_{z ( v)})}$ be $\Id$, otherwise let $X^v_{\alpha ( e^v_1) \ldots \alpha ( e^v_{z ( v)})}$ be $0$. By construction, the contraction of these tensors really gives $\tilde{\rho}$.

As the index used at the coding, $\alpha=(\alpha_1,\alpha_2)$, can only take $B_1B_2=poly(K)$ different values, the above coding is a PEPO with $poly(K)$  (equivalently $poly(N)$) bond dimension. Thus, for any $\beta<\beta^*$, we gave an efficient PEPO description of the Gibbs state. Moreover, Theorem (\ref{thm:kliesch}) holds for $\beta' = \beta / 2 M$ instead of $\beta$ if the trace norm is replaced by $\| . \|_{2 M}$ without any essential modification. Therefore, by taking $M$ such that $\beta'<\beta^*$, that is, $M=O(\beta)$, this result can be extended to lower (but finite) temperatures as well (see Proposition \ref{prop:T}). However, after this step, the approximating operator will be a PEPO exponentiated $M$ times. Therefore the bond dimension required for the PEPO description of the Gibbs state at arbitrary temperature is $N^{O(\beta)}$.

\section{Summary and Outlook}

We have analyzed the ability of tensor networks to describe thermal
(Gibbs) equilibrium states of lattice Hamiltonians with local interactions.
First, using a Trotter expansion and a compression method, we have
shown that it is possible to approximate that state with a PEPO whose
bond dimension scales as $N^{O(\log N)}$, where $N$ is the system size
(number of vertices in the lattice). This result is valid for
any finite temperature and spatial dimension. It also holds true at
zero temperature as long as the Hamiltonian is gapped and the density
of states for any energy interval only grows polynomially with the system
size. Second, building on Hastings' construction \cite{Hastings2006}, we
have shown that it is possible to find a PEPO with a $poly(N)$ bond dimension
at any finite temperature and spatial dimension.

There are some straightforward implications of the results derived
here. First, even though we have concentrated on PEPOs, it is trivial to
express our results in terms of (pure) PEPS. At finite temperature, we
can just consider the PEPO corresponding to half the temperature, and apply it
to locally maximally entangled states in order to obtain a purification
in terms of a PEPS with a polynomially growing bond dimension \cite{Verstraete2004b}.
At zero temperature, we can simply apply the constructed PEPO to a random
product state in order to show that there exists a PEPS with $D=N^{O(\log N)}$.
Second, for translationally invariant problems in regular lattices, our
construction may break  translational invariance (as we select some order of the bonds).
But it is always possible \cite{Perez-Garcia2007} to make a PEPO (or PEPS)
translationally invariant with an increase of the bond dimension by just
a factor of $N$. Third, even though we have considered Hamiltonians
interacting along the edges in the graph, our construction can be
easily extended to the case in which the local Hamiltonians act on
plaquettes. The idea is that at the Trotter decomposition we have made no assumption on the support of the individual Hamiltonian terms, whereas at the coding procedure, we still need to keep information contained in a constant number of columns: in an edge $e=(v,w)$, we can keep the information contained in the columns corresponding to Hamiltonian terms that act non-trivially on either $v$ or $w$. In such a coding the same piece of information is specified in more than one edge, but their consistency can be checked locally, at the vertices. The cluster expansion technique can be applied with no essential modification as the number of terms acting on the boundary of a cluster can still be upper bounded by a constant times the size of the cluster, and the number of clusters containing  $l$ terms is still bounded by $\gamma^l$, where $\gamma$ is a lattice growth constant \cite{Klarner1967}.
Finally, our construction can also be straightforwardly extended to
fermions with the result that we just have to use fermionic PEPS
\cite{Kraus2010}.

\begin{acknowledgements}
	We thank M.C. Ba\~nuls, T. Vidick, Z. Landau and U. Vazirani for discussions. This work has been partially supported by the EU project SIQS, and the DFG project NIM. We also thank the Benasque Center for Science, the Perimeter Institute (Waterloo), and the Simon's Center for the Theory of Computing (Berkeley), where part of the work was carried out, for their hospitailty. NS acknowledges the support from the Alexander von Humboldt foundation and the EU project QALGO. JIC acknowledges support from the Miller Institute in Berkeley.
\end{acknowledgements}

\appendix

\section{Proof of Proposition \ref{prop:T}}
Here we present the proof of Proposition \ref{prop:T}. The proof consists of two steps. First, by the positivity of $\eta$, we show that if $\ep<1/3$ and
   \begin{equation}
     \| \eta - \tau \|_{2 M} \leq \frac{\varepsilon}{M} \| \eta \|_{2 M},
   \end{equation}
then 
\begin{equation}
    \| \eta^2 -  \tau^{\dagger}\tau \|_M \leq 3 \varepsilon \| \eta^2 \|^{}_M .
\end{equation}
Using the identity $a^2-b^2=a(a-b)+(a-b)b$ and the triangle inequality we obtain
   \begin{equation}
     \| \eta^2 -\tau^{\dagger} \tau \|_M = \| \eta ( \eta -\tau)\|_{M} + \|( \eta - \tau^{\dagger}) \tau \|_M ,
   \end{equation}
   thus using the H\"older inequality and that $\| X \|_{2 M} = \| X^{\dagger} \|_{2 M}$ , we conclude that
   \begin{equation}
     \| \eta^2 - \tau\tau^{\dagger} \|_M \leq \left(\| \eta\|_{2M} + \|\tau\|_{2M}\right) \| \eta -\tau\|_{2M} .
   \end{equation}   
$\| \eta -\tau\|_{2M}$ is bounded by the assumptions of the statement, and so is $\|\tau\|_{2M}$ by the triangle inequality. Therefore
   \begin{equation}
     \| \eta^2 - \tau\tau^{\dagger} \|_M \leq (2+\frac{\varepsilon}{M})\frac{\varepsilon}{M}\|\eta\|_{2M}^2.
   \end{equation}
$\eta$ is positive, thus $\|\eta\|_{2M}^2=\|\eta^2\|_{M}$. If $\varepsilon<1$, then
      \begin{equation}\label{eq:diff2}
        \| \eta^2 - \tau\tau^{\dagger} \|_M \leq 3\frac{\varepsilon}{M}\|\eta^2\|_{M}.
      \end{equation}
This completes the proof of the first step. 

Second, we prove that if Eq. (\ref{eq:diff2}) holds, then 
  $$\| \eta^{2M} - (\tau^\dagger \tau)^{M} \|_1 \leq 9\varepsilon
     \| \eta^{2M} \|_1 . $$
     The proof is basically the same as that of the first step. Using the identity $a^m-b^m=\sum_i a^i (a-b) b^{m-i-1}$ and the triangle inequality we obtain
  \begin{equation}
  \| \eta^{2M} - (\tau^\dagger \tau)^M \|_1 \leq \sum_{i = 0}^{M - 1} \| \eta^{2i} (
       \eta^2 - \tau^\dagger \tau) (\tau^\dagger \tau)^{M - i - 1} \|_1 
  \end{equation}
  Hence by H\"older's inequality the difference is upper bounded by 
  \begin{equation}\label{eq:expand_M}
         \sum_{i = 0}^{M - 1} \left\| \eta^{2i} \right\|_{\frac{M}{i}}  \left\| \eta^2 - (\tau^\dagger\tau) \right\|_M  \left\| (\tau^\dagger\tau)^{M-i-1} \right\|_{\frac{M}{( M-i-1)}}.
  \end{equation}
  For $X$ positive semidefinite, and any  real number $r$, $X^{r} \geq 0$,  and thus by the definition of the Schatten norms
\begin{equation}\label{eq:schatten_M}
   \| X^{r} \|_{M / r} = \| X \|_M^{r}.
\end{equation}  
Applying (\ref{eq:schatten_M}) to $\eta$ and $\tau^\dagger\tau$ in (\ref{eq:expand_M}), the inequality takes the following form: 
  \begin{equation}
    \| \eta^{2M} - (\tau^\dagger\tau)^M \|_1 \leq \sum_{i = 0}^{M - 1} \| \eta^2 \|^{i}_M \| \eta^2 - \tau^\dagger\tau \|_M  \| \tau^\dagger\tau \|^{M - i - 1}_M 
  \end{equation}
  $ \| \eta^2 - \tau^\dagger\tau \|_M$ is bounded by Eq. (\ref{eq:diff2}). Hence, by the triangle inequality, $\|\tau^\dagger \tau\|_M$ is bounded as well, 
  $$\| \tau^\dagger\tau \|_M
    \leq \| \eta^2 \|_M + \| \eta^2 - \tau^\dagger\tau \|_M \leq \left( 1 +
    \frac{3\varepsilon}{M} \right) \| \eta^2 \|_M .$$
As $1+3\varepsilon/M>1$, we can upper bound the sum by taking $(1+3\varepsilon/M)^M$ as common factor in every term,
  \begin{equation}
     \| \eta^{2M} - (\tau^\dagger\tau)^M \|_1  \leq M \frac{\varepsilon}{M} \left( 1 + \frac{3\varepsilon}{M} \right)^M \| \eta^2 \|_M^M .
  \end{equation}
Since $(1+3\varepsilon/M)^M<e^{3\varepsilon}<3$, if $\epsilon\leq 1/3$, the statement of the Proposition follows:
  \begin{equation}
     \| \rho^{2M} - (\tau^\dagger\tau)^M \|_1  \leq 9\varepsilon\| \eta^2 \|_M^M .
  \end{equation}

\section{Proof of lemma \ref{lem:trotterproof}} \label{sec:appB}

Here we present the proof of Lemma \ref{lem:trotterproof}, and derive how big the number of Trotter steps should be chosen for a good approximation of the Gibbs state. The proof relies on the fact that if $M$ is big enough, then both $\eta$ and $\tau$ are close to $\Id-\beta H /2M$.
 
  By the use of H{\"o}lder's inequality we obtain
  \begin{equation}
    \| \eta - \tau \|_{2 M} = \|  \eta^{- 1} \eta ( \eta - \tau) \|_{2 M} \leq
    \| \eta \|_{2 M} \| \eta^{- 1} \|  \| \eta - \tau \| \nonumber.
  \end{equation}
 The norm of $\eta^{-1}$ can be upper bounded by a constant if $M >\beta K/2$:
  \begin{equation}
    \| \eta^{- 1} \| \leq e^{\beta K / 2 M} \leq 3.
  \end{equation}
The norm of $\eta -\tau $ will be bounded with the help of the triangle inequality, by adding and subtracting $\Id- \beta H/2M $:
  \begin{equation} \label{eq:trottertriangle}
  	\left\| \eta - \tau \right\| \leq \left\| \eta - \Id + \frac{\beta H}{2 M} \right\|  + \left\| \tau - \Id + \frac{\beta H}{2 M} \right\|.
  \end{equation}
 We will use the following bound on the Taylor expansion of the exponential function to upper bound these expressions.
\begin{lemma} \label{lem:expapprox} The following two bounds hold:
$$\norm{e^A-\id} \leq \norm{A}e^{\norm{A}},$$
$$\norm{e^A-\id -A}\leq \frac{\norm{A}^2}{2}e^{\norm{A}}.$$
\end{lemma}

\begin{proof}
$e^A=\sum_n \frac{A^n}{n!}$, thus
\begin{equation}
e^A-\sum_{n=0}^{k}\frac{A^n}{n!}=\sum_{n=k+1}^{\infty}\frac{A^n}{n!}.
\end{equation}
Therefore the norm of the difference can be upper bounded by  the triangle inequality
\begin{equation}
\|e^A-\sum_{n=0}^{k}\frac{A^n}{n!}\|\leq\sum_{n=k+1}^{\infty}\frac{\norm{A}^n}{n!}\leq\frac{\norm{A}^{k+1}}{(k+1)!}\sum_{n=0}^{\infty}\frac{\norm{A}^n}{n!},
\end{equation}
since $(n+k+1)!\geq n! (k+1)!$. Summing up we have the desired inequality
\begin{equation}
\|e^A-\sum_{n=0}^{k}\frac{A^n}{n!}\|\leq \frac{\norm{A}^{k+1}}{{k+1}!}e^{\norm{A}}.
\end{equation}
 The statements correspond to the particular cases $k=0,1$.
\end{proof}

Due to the previous lemma, we can bound the first part of the right hand side of Eq. \ref{eq:trottertriangle}:
\begin{equation}\label{eq:b6}
      \left\| e^{- \frac{\beta H}{2 M}} - 1 + \frac{\beta
     H}{2 M} \right\| \leq \frac{\beta^2 K^2}{8M^2} e^{\beta K/2M}\leq \frac{\ep}{2M}.
\end{equation}
If $M\geq \beta^2 K^2/\epsilon$ and $M\geq \frac{\beta K}{2}$, because then $e^{\beta K/2M}\leq 3$ and
$3/8\leq 1/2$.

The second part of the right hand side of Eq. \ref{eq:trottertriangle} can be written as
\begin{equation}\label{eq:difference_tau}
      \tau-\Id+ \frac{\beta H}{2 M} =  \prod_i \left[ \Id +x_i \right] - \Id + \frac{\beta H}{2 M} ,
\end{equation}
where $x_i$ is as in Eq. (\ref{eq:expansion}). Let us expand the product. The zeroth term cancels  out, whereas the 1st order term is 
$$\sum_i x_i + \frac{\beta H}{2M}.$$
The norm of the sum of the $k^{th}$ order terms can be upper bounded by 
$${K \choose k} \left(\frac{\beta}{2M} e^\frac{\beta}{2M}\right)^k\leq \left(\frac{3\beta K}{2M}\right)^k$$ 
if $M\geq \beta/2$, because there are $K \choose k$ $k^{th}$ order terms, and the norm of $\|x_i\|$ can be bounded by lemma \ref{lem:expapprox}: 
$$\|x_i\|\leq \frac{\beta}{2M} e^\frac{\beta}{2M}.$$
Therefore, after expanding the product in Eq. (\ref{eq:difference_tau}), we obtain that 
\begin{equation}
   \left\| \tau - 1 +
     \frac{\beta H}{2 M} \right\| \leq \left\| \sum_i x_i + \frac{\beta h_i}{2M}\right\| + \sum_{k=2}^{\infty} \left(\frac{3 K\beta }{2M}\right)^k.
\end{equation}
  The first term can be again bounded by Lemma \ref{lem:expapprox} as  $x_i+\beta h_i/2M=e^{-\beta h_i/2M}-\Id+\beta h_i/2M$:
  \begin{equation}\label{eq:b9}
   \left\|\sum_i x_i-\frac{\beta h_i}{2M}\right\| \leq K \frac{\beta^2 }{4M^2}e^\frac{\beta}{2M}\leq  \frac{K\beta^2 }{M^2},
  \end{equation}
since if $M>\beta/2$, then $e^\frac{\beta}{2M}<4$. The second term can be upper bounded by
\begin{equation}\label{eq:b10}
\sum_{k=2}^{\infty} \left(\frac{3 K\beta }{2M}\right)^k= \left(\frac{3 K\beta }{2M}\right)^2\frac{1}{1-\frac{3 K\beta }{2M}}\leq \frac{5 K^2\beta^2 }{M^2},
\end{equation}
since if $M>3K\beta$, then $\frac{1}{1-\frac{3 K\beta }{2M}}\leq 2$, and $9/2\leq 5$. Finally,  $K>1$ and thus the sum of the bounds obtained in Eq. (\ref{eq:b9}) and in Eq. (\ref{eq:b10}) can be upper bounded by
\begin{equation}\label{eq:b11}
\left\| \tau - 1 + \frac{\beta H}{2 M} \right\| \leq \frac{6 K^2\beta^2 }{M^2}\leq \frac{\ep}{2M},
\end{equation}
if $M>12K^2\beta^2 \frac{1}{\ep}$. 

Putting together the two bounds in Eq. (\ref{eq:b6}) and Eq. (\ref{eq:b11}), we obtain that 
\begin{equation}
\|\eta -\tau\|\leq \frac{\epsilon}{M}
\end{equation}
if $M>12K^2\beta^2 /\ep$. Therefore, the statement follows: if $M>36K^2\beta^2 /\ep$, then
\begin{equation}
\|\eta -\tau\|_{2M}\leq \frac{\epsilon}{M} \|\eta\|_{2M}.
\end{equation}

\section{Proof of the M{\"o}bius inversion}

Here we prove the Lemma \ref{lem:Mobius}. 
The first part of the statement is that $\check{\hat{f}}=f$. Let us define $f'(I)$ as
\begin{equation}
  f' ( \mathcal{I}) = \sum_{\mathcal{J} \supseteq \mathcal{I}} ( - 1)^{| \mathcal{J}\backslash \mathcal{I} |}  \hat{f} ( \mathcal{J}).
\end{equation}
Then, the statement is that $f'(\mathcal{I})=f(\mathcal{I})$. Let us express $\hat{f}$ with the help of $f$ as in Eq. (\ref{eq:hat}):
\begin{equation}
  f' (\mathcal{I}) =  \sum_{\mathcal{J} \supseteq \mathcal{I}} ( - 1)^{| \mathcal{J}\backslash \mathcal{I} |}  \sum_{\mathcal{K} \supseteq \mathcal{J}} f ( \mathcal{K}).
\end{equation}
By changing the order of the sums we obtain
\begin{equation}
  f' ( \mathcal{I}) = \sum_{\mathcal{K} \supseteq \mathcal{I}} f ( \mathcal{K}) \sum_{\mathcal{J} : \mathcal{K} \supseteq \mathcal{J} \supseteq \mathcal{I}} (
  - 1)^{| \mathcal{J}\backslash \mathcal{I} |}.
\end{equation}
We evaluate now the second sum. Suppose first $\mathcal{K} \neq \mathcal{I}$: then
\[ \sum_{\mathcal{J} : \mathcal{K} \supseteq \mathcal{J} \supseteq \mathcal{I}} ( - 1)^{| \mathcal{J}\backslash \mathcal{I} |} = \sum_{\mathcal{J}'
   \subseteq \mathcal{K}\backslash \mathcal{I}} ( - 1)^{| \mathcal{J}' |} = ( 1 - 1)^{| \mathcal{K}\backslash \mathcal{I} |} = 0.
\]
Otherwise, if $\mathcal{K} = \mathcal{I}$, then the sum is one. Substituting this back in the expression of $f'(\mathcal{I})$, we get 
\begin{equation}
  f' ( \mathcal{I}) = \sum_{\mathcal{K} \supseteq \mathcal{I}} f ( \mathcal{K}) \sum_{\mathcal{J} : \mathcal{K} \supseteq \mathcal{J} \supseteq \mathcal{I}} (
  - 1)^{| \mathcal{J}\backslash \mathcal{I} |} = f (\mathcal{I}).
\end{equation}
This proves the first part of the statement. The second part, $\hat{\check{f}}=f$, works similarly. Let us define now $f''$ as follows:
\begin{equation}
  f'' ( \mathcal{I}) = \sum_{\mathcal{J} \subseteq \mathcal{I}} \check{f} (\mathcal{J})  .
\end{equation}
Thus, we have to prove that $f''=f$. Substituting back the expression for $\check{f}$ (as in Eq. \ref{eq:check}) in this equation, we obtain
\begin{equation}
  f'' ( \mathcal{I}) = \sum_{\mathcal{J} \subseteq \mathcal{I}}  \sum_{\mathcal{K} \subseteq \mathcal{J}} (-1)^{| \mathcal{J}\backslash \mathcal{K} |} f ( \mathcal{K}) .
\end{equation}
 By changing the order of the two sums we obtain
\begin{equation}
  f'' ( \mathcal{I}) =  \sum_{\mathcal{K} \subseteq \mathcal{I}} f
  ( \mathcal{K}) \sum_{\mathcal{J} : \mathcal{K} \subseteq \mathcal{J} \subseteq \mathcal{I}} ( - 1)^{| \mathcal{J}\backslash \mathcal{K} |}.
\end{equation}
The second sum is again $\delta_{\mathcal{K}, \mathcal{I}}$, and thus
\begin{equation}
  f'' ( \mathcal{I}) =f(\mathcal{I}).
\end{equation}

\section{Proof of Lemma \ref{lem:S_norm}} \label{sec:appD}
Here we present the proof of Lemma \ref{lem:S_norm}. In $\mu$ two types of terms occur. First, if $i$ refers to a column that has been summed up ($i\notin \mathcal{I}$), then in every row of that column the term $e^{-\beta \tilde{h}_i}$ appears. Second, if $i\in \mathcal{I}$, then the sum on that column has not been evaluated, therefore the corresponding term in row $j$ is $x_i^{\lambda_{i,j}}$. We now separate these terms:
\begin{equation}\label{eq:mu_nice}
	\mu = \prod^M_{j = 1} \prod_{i \in \mathcal{I}} y_i^{\lambda_{i, j}}  \prod_{i\in \{1..2K\}} e^{- \beta \tilde{h}_i / 2 M},
\end{equation}
where we have introduced
\begin{equation}
  y_i = \prod_{j < i} e^{- \beta \tilde{h}_j / 2 M} x_i \prod_{j < i} e^{\beta
  \tilde{h}_j / 2 M} \cdot e^{\beta \tilde{h}_i / 2 M}.
\end{equation}
The norm of $y_i$ can be bounded by the norm of $x_i$ as follows:
\begin{equation}
  \| y_i \| \leq \| x_i \| e^{\beta K / M}
\end{equation}
 because $\|e^{- \beta \tilde{h}_j/2M}\|\leq 1$ and  $\|e^{\beta \tilde{h}_i / 2 M}\|\leq  e^{\beta/ 2 M}$ and there are at most $2K$ such terms in $y_i$. Thus, by applying Lemma \ref{lem:expapprox} to $x_i$, we obtain:
\begin{equation}\label{eq:bound_y}
  \| y_i \| \leq \| x_i \| e^{\beta K / M} \leq \frac{\beta}{2 M}
  e^{\beta / 2 M} e^{\beta K / M} \leq \frac{3 \beta}{2 M} ,
\end{equation}
since $e^{\beta (2K+1) / 2M}\leq 3$ if $M>2\beta K >\beta( K+1/2)$. We now apply H{\"o}lder's inequality to Eq. (\ref{eq:mu_nice}) in order to bound $\|\mu\|_1$:
\[ \left\| \mu \right\|_1 \leq  \prod_i \| y_i
   \|^{\sum_j \lambda_{i, j}}  \left\| \prod_{i\in \{1..2K\}} e^{- \beta \tilde{h}_i / 2 M}
   \right\|_{ M}^{ M} .\]
The last expression of the right hand side is $(\tau^\dagger \tau)$ from the Trotter expansion formula.  By the use of an other H{\"o}lder's inequality 
\[\|\tau^\dagger \tau\|_{M}\leq \|\tau\|_{2M}^2.\]
Using the triangle inequality and Lemma \ref{lem:trotterproof} with the choice $\epsilon=1/2$, we obtain that
\begin{equation}
\left\| \tau \right\|_{2 M} \leq \left( 1 + \frac{1}{2M} \right) \| e^{- \beta
H / 2 M} \|_{2 M}
\end{equation}
if $M >  72 \beta^2 K^2$; therefore
\begin{equation}
  \|\mu\|_1 \leq \prod_i \| y_i \|^{k_i}  \left( 1 + \frac{1}{2M} \right)^{2
  M} \| e^{- \beta H / 2 M} \|_{2 M}^{2 M} .
\end{equation}
Using the bound (\ref{eq:bound_y}) on $\|y_i\|$, and the fact that $(1+1/2M)^{2M}<e<3$, we obtain the statement of the lemma,
\begin{equation}
  \|\mu\|_1 \leq 3 \left(\frac{3\beta}{M}\right)^{k_1+\dots k_n} \|e^{-\beta H}\|_1.
\end{equation}

\section{Lemma on the sum of binomial coefficients} \label{sec:appE}
We need the following lemma to upper bound a sum of binomial coefficients in Eq. (\ref{eq:binom_lem_need}):

\begin{lemma}{\label{lem:binom}}
  $\sum_{k > L} \binom{M}{k} x^k \leq e^{Mx} \left( \frac{eMx}{L}
  \right)^L$.
\end{lemma}

\begin{proof}
First, as ${M \choose k}\leq M^k/k!$, we have
\begin{equation}
\sum_{k > L} \binom{M}{k} x^k \leq \sum_{k \geq L} \frac{1}{k!} (Mx)^k.
\end{equation}
We then use $(L+n)!\geq n! L!$ and sum up over $n=k-L$.
\begin{equation}
\sum_{k  \geq L} \binom{M}{k} x^k \leq \frac{1}{L!} (Mx)^L e^{Mx}.
\end{equation}
Finally, by Stirling's formula, we have the desired result:
\begin{equation}
\sum_{k > L} \binom{M}{k} x^k \leq e^{Mx} \left( \frac{eMx}{L}.
  \right)^L
\end{equation}
\end{proof}

\section{On the cluster expansion }\label{sec:appF}

In this Section we show how to use the M{\"o}bius inversion to reproduce the cluster expansion.
In particular, we show that 
\begin{equation}
  g ( \mathcal{I}) = \sum_{\tmscript{\begin{array}{c}
    w \in \mathcal{I}^{*}\\
    \tmop{supp} ( w) = \mathcal{I}
  \end{array}}} \frac{( - \beta')^{| w |}}{| w | !} h_w \label{eq:fcheck}
\end{equation}
is the (inverse) M\"obius transform\footnote{Note that the definition of M{\"o}bius inversion is slightly different in this context: we use a sum over $\mathcal{J} \subseteq \mathcal{I}$ instead of a sum over $\mathcal{J} \supseteq \mathcal{I}$. The inverse is defined likewise.} of

\begin{equation}
  f ( \mathcal{J}) = e^{- \beta' H ( \mathcal{J})}.
\end{equation}
Let us consider the M\"obius transform of $g$:
\begin{equation}
\hat{g}(\mathcal{I}) = \sum_{\mathcal{J} \subseteq \mathcal{I}} g ( \mathcal{J}) = \sum_{\mathcal{J} \subseteq \mathcal{I}} \sum_{\tmscript{\begin{array}{c}
    w \in \mathcal{J}^{*}\\
    \tmop{supp} ( w) = \mathcal{J}
  \end{array}}} \frac{( - \beta')^{| w |}}{| w | !}h_w.
\end{equation}
This means that in $\hat{g}$ we have to sum up for all words in $\mathcal{I}^*$. Indeed, in the sum every word is counted exactly once as we sum up all possible supports. This implies that
\begin{equation}
\hat{g}(\mathcal{I}) = e^{-\beta' H(\mathcal{I})}=f(\mathcal{I}),
\end{equation}
and therefore by the M\"obius inversion formula $g=\check{f}$.

We now show that obtaining the form Eq. (\ref{eq:cluster_approx}) of $\tilde{\rho}$ is much easier with these tools. The proof follows from the multiplicativity of $\check{f}$: if $\mathcal{I}$ and $\mathcal{J}$ are non-overlapping clusters, then  $\check{f} ( \mathcal{I} \cup \mathcal{J}) = \check{f} ( \mathcal{I})  \check{f} ( \mathcal{J})$. Indeed,
\begin{equation}
  \check{f} ( \mathcal{I} \cup \mathcal{J}) = \sum_{\mathcal{K} \subseteq \mathcal{I} \cup \mathcal{J}} ( - 1)^{| \mathcal{I}\cup \mathcal{J} \backslash \mathcal{K} |}  e^{- \beta H ( \mathcal{K} )} ,
\end{equation}
where we have used the multiplicativity of the exponential. $\mathcal{K}$ can be broken into two parts: $\mathcal{K}_\mathcal{I}=\mathcal{K}\cap \mathcal{I}$ and $\mathcal{K}_\mathcal{J}=\mathcal{K}\cap \mathcal{J}$. Then both the $-1$ factor and the exponential factorizes as follows:
\begin{equation}
  \check{f} ( \mathcal{I} \cup \mathcal{J}) = \mkern-18mu \sum_{\tmscript {\begin{array}{c} \mathcal{K}_\mathcal{I}\subseteq \mathcal{I}\\ \mathcal{K}_\mathcal{J}\subseteq \mathcal{J}\end{array}}} \mkern-18mu ( - 1)^{| \mathcal{I}\backslash \mathcal{K}_\mathcal{I} |}  ( - 1)^{| \mathcal{J}\backslash \mathcal{K}_\mathcal{J} |} e^{- \beta H ( \mathcal{K}_\mathcal{I})} e^{- \beta H ( \mathcal{K}_\mathcal{J})}
\end{equation}
and this sum is nothing but $\check{f} ( \mathcal{I}) \cdot \check{f} ( \mathcal{J})$.
This implies that the Gibbs state admits the following form: 
\begin{equation}
  \tilde{\rho} = \sum_{\mathcal{I} \subseteq \mathcal{E}} \prod_{\tmscript{\begin{array}{c}
      i : \mathcal{I}_i \text{ are the}\\
       \text{clusters in } \mathcal{I}
    \end{array}}} \check{f} ( \mathcal{I}_i),
\end{equation}
and thus the approximation $\tilde{\rho}$ is nothing but
\begin{equation}
  \tilde{\rho} = \sum_{ \tmscript{\begin{array}{c}
    \mathcal{I} \in \mathcal{C}_L\\
    \mathcal{I} = \uplus \mathcal{I}_i
  \end{array}}} \prod_i \check{f} ( \mathcal{I}_i).
\end{equation}
 as in Eq. (\ref{eq:cluster_approx}).
\bibliography{citations.bib}

\end{document}